\documentclass[letter,11pt]{article}
\usepackage{amsmath, amsthm, amssymb}
\usepackage{graphicx}
\usepackage[latin1]{inputenc}
\usepackage{fullpage}
\usepackage[colorlinks,linkcolor=blue,citecolor=blue,urlcolor=magenta,linktocpage,plainpages=false]{hyperref}
\usepackage{xcolor}
\usepackage{algorithm}
\usepackage[noend]{algpseudocode}
\usepackage{empheq}
\usepackage{bm}
\usepackage{times}
\usepackage{xspace}
\usepackage[b]{esvect}
\usepackage{thmtools}

\usepackage[T1]{fontenc}

\declaretheoremstyle[
    shaded={bgcolor=\color{rgb}{0.9,0.9,0.9}}  
]{theorem}

\newtheorem{lemma}{Lemma}
\newtheorem{thm}{Theorem}
\newtheorem{cor}{Corollary}
\newtheorem{ass}{Assumption}

\newcounter{mynotes}
\newcommand{\mnote}[1]{}

\newcommand{\D}{\mathcal{D}}
\newcommand{\ip}[2]{\langle #1, #2\rangle}
\newcommand{\ipbig}[2]{\bigg\langle #1\,,\, #2\bigg\rangle}
\newcommand{\R}{\mathbb{R}}

\newcommand{\E}{\mathbb{E}}
\newcommand{\F}{\mathcal{F}}
\newcommand{\OPT}{\textsc{OPT}\xspace}
\newcommand{\vOPT}{\vv{\textsc{OPT}}}
\newcommand{\ones}{\bm{1}}
\newcommand{\OCOalg}{\textsc{ss-ftrl}\xspace}

\renewcommand{\d}{\textrm{d}}
\newcommand{\ocp}{\textsc{OnlineCvx}\xspace}
\newcommand{\welf}{\textsc{OnlineWelfare}\xspace}
\newcommand{\mixed}{\textsc{Mixed}\xspace}

\newcommand{\e}{\varepsilon}

\newcommand{\alg}{\textsc{Alg}\xspace}
\newcommand{\algw}{\textsc{wAlg}\xspace}
\newcommand{\vOPTr}{\vOPT_{Stoch}^{\,rec}}
\newcommand{\wtilde}[1]{\widetilde{#1}}

\newcommand{\blue}[1]{{\color{blue} #1}}

\DeclareMathOperator{\argmax}{argmax}
\DeclareMathOperator{\argmin}{argmin}

\title{Robust Algorithms for Online Convex Problems \\ via Primal-Dual}
\author{Marco Molinaro}
\date{}

\begin{document}
	\maketitle

	\begin{abstract}
	
		The importance of primal-dual methods in online optimization can hardly be overstated, and they give several of the state-of-the art results in both of the most common models for online algorithms: the \emph{adversarial} and the \emph{stochastic/random order} models. Here we try to provide a more unified analysis of primal-dual algorithms to better understand the mechanisms behind this important method. With this we are able of recover and extend in one goal several results of the literature. 
		
		In particular we obtain \textbf{robust} online algorithm for fairly general online convex problems: we consider the \mixed model where in some of the time steps the data is stochastic and in the others the data is adversarial. Both the quantity and location of the adversarial time steps are unknown to the algorithm. The guarantees of our algorithms interpolate between the (close to) best guarantees for each of the pure models. In particular, the presence of adversarial times does not degrade the guarantee relative to the stochastic part of the instance. 
		
		More concretely, we first consider \emph{online convex programming}: in each time step a feasible set $V_t$ is revealed, and the algorithm needs to
select $v_t \in V_t$ to minimize the total cost $\psi(\sum_t v_t)$, for a convex function $\psi$. Our robust primal-dual algorithm for this problem on the \mixed model recovers and extends, for example, a result of Gupta et al.~\cite{raviConvexPD} as well as the recent work on $\ell_p$-norm load balancing~\cite{molinaro}. We also consider the problem of \emph{welfare maximization with convex production costs}: in each time a customer presents a value $c_t$ and resource consumption vector $a_t$, and the goal is to fractionally select customers to maximize the profit $\sum_t c_t x_t - \psi(\sum_t a_t x_t)$. Our robust primal-dual algorithm for this problem on the \mixed model recovers and extends the result of Azar et al.~\cite{cvxPDFOCS}.

	Given the ubiquity of primal-dual algorithms, we hope that the ideas of the analyses presented here will be useful in obtaining other robust algorithm in the \mixed or related models. 

	\end{abstract}

	\section{Introduction}

	The importance of primal-dual methods in online optimization can hardly be overstated. For example, several of the results and applications of Online Learning can be seen as solving a convex-concave game using a primal-dual procedure. In online algorithms, the focus of this work, primal-dual algorithms give several of the state-of-the art results in both of the most common online models: the \emph{adversarial} and the \emph{stochastic/random order} models. 
	
	Aiming at a better understand the mechanisms behind this important method, we provide in this paper a conceptually simpler and more unified analysis of primal-dual algorithms for fairly general convex online problems, and are able of recover and extend in one goal several results of the literature. 
	
	As a consequence of this unified analysis, we also obtain new \textbf{robust} online algorithms by considering the \mixed online model that interpolates between the adversarial and stochastic models: in some time steps the data revealed online is stochastic, drawn from an unknown distribution, and in the other time steps the data is adversarial (we will clarify details shortly). Both the quantity and location of the adversarial time steps are unknown to the algorithm. We obtain robust primal-dual algorithms in the \mixed model whose costs/profits are comparable to $\alpha \OPT_{Stoch} + \beta \OPT_{Adv}$, where $\OPT_{Adv}$ is the optimal offline solution over the adversarial part of the input, and $\OPT_{Stoch}$ is defined similarly, and $\alpha,\beta$ are approximation factors that match or almost match the best known ones for the respective pure models. Thus, these algorithms are robust: the presence of adversarial data does not destroy the stronger stochastic guarantee.
	
	While our focus is on better understanding how to think about and analyze primal-dual algorithms, we believe that these new robust algorithms are interesting in their own right, and bridging between the (optimistic) stochastic and (pessimistic) adversarial models has been a topic of significant interest in both online algorithms~\cite{meyerson,adSim,welfareSim,molinaro,kesselheimMolinaro,guptaSahilRobust,KKN15,mixedModelSpikes} and online learning (see \cite{guptaCOLT19} and references within). Many of the classical algorithms for sequential decision making are quite brittle and heavily dependent on the model; for example, it is easy to see that the classic threshold-based algorithm for the Secretary Problem obtains an arbitrarily bad solution even in the presence of a single adversarial item (see~\cite{kesselheimMolinaro,guptaSahilRobust} for further discussion). This highlights the importance of algorithm that are robust by design, and indicates that this may be an important factor that contributes to practical performance. Finally, given the ubiquity of primal-dual algorithms, we hope that the ideas of the analyses presented here will be useful in obtaining other robust algorithm in the \mixed or related models.

	
	\subsection{Our results} \label{sec:results}
	
	 
	\paragraph{Online Convex Programming.} We first consider the following general problem: A convex objective function $\psi$ is known upfront. At time $t$, a feasible region $V_t \subseteq [0,1]^m$ arrives and using only the information seen thus far the algorithm needs to choose a feasible point $v_t \in V_t$. The goal is to minimize the total cost $\psi(\sum_{t=1}^n v_t)$. Notice that no additional properties are imposed on the feasible sets, such as convexity or being of packing/covering-type.\footnote{We will assume that the feasible sets $V_t$ are presented through an exact optimization oracle.} We use \ocp to denote this problem. 
	
		
	The \mixed input model is instantiated in this context as follows. The set of time steps $[n]$ is partitioned arbitrarily into the adversarial times $Adv \subseteq [n]$ and the stochastic times $Stoch \subseteq [n]$. Both the quantity and the location of the adversarial times are unknown to the algorithm. The feasible sets $V_t$ are generated adversarially for times $t \in Adv$. In addition, there is an unknown distribution $\mathcal{D}$ over subsets of $[0,1]^m$, and in each stochastic time $t$ the feasible set $V_t$ is sampled independently from $\D$. To keep the model clean, we assume that the adversary is non-adaptive, so in particular the partition of time steps and the adversarial part of the sequence do not depend on the draws of the stochastic times. 
	
	We define $\OPT_{Adv} := \min \{\psi(\sum_{t \in Adv} v_t) : v_t \in V_t, ~\forall t \in Adv\}$ as the offline optimum over only the adversarial part of the input. The optimum over the stochastic part, $\OPT_{Stoch}$, is defined similarly: let $v^*$ be the function that maps each set $V \subseteq [0,1]^m$ to a point $v^*(V) \in V$ (i.e., a feasible selection) and that minimizes $\E \psi(\sum_{t \in Stoch} v^*(V_t))$, and define $\OPT_{Stoch} := \E \psi( \sum_{t \in Stoch} v^*(V_t))$.\footnote{If the minimization is not achieved by any selector $v^*$, we define $\OPT_{Stoch}$ by taking the infimum over them.}
	
	
	Our first concrete result is an algorithm for \ocp on the \mixed model with expected cost roughly at most $\alpha \OPT_{Stoch} + \beta \OPT_{Adv}$. As in the literature for related problems (e.g., online set cover with convex costs~\cite{cvxPDFOCS}) we focus on the case where the cost function has non-decreasing gradients,\footnote{This means that if a vector $u$ is coordinate-wise smaller than $v$, then $\nabla \psi(u) \le \nabla \psi(v)$.} which models diseconomies of scale~\cite{onlinePDICALP2}, and the approximation factors depend on the growth rate of the cost function. We say that a function $\psi : \R^m_+ \rightarrow \R_+$ has \emph{growth of order at most $p$} if for every vector $u$ and $\gamma \ge 1$ we have $\nabla \psi(\gamma u) \le \gamma^{p-1} \nabla \psi(u)$.	Some functions that have growth rate at most $p$ are $\psi(u) = (\sum_i u_i)^p$, the $\ell_p$-norm raised to the power $p$, and more generally convex polynomials of degree $p$ with non-negative coefficients. 
	
	
	\begin{thm} \label{thm:main}
		Consider the \ocp problem in the \mixed model. Suppose the objective function $\psi : \R^d_+ \rightarrow \R_+$ is convex, differentiable, non-decreasing, has $\psi(0) = 0$, and $\nabla \psi$ is non-decreasing. If $\psi$ has growth of order at most $p \ge 2$, then Algorithm \ref{alg:cvxProg} has expected cost at most
		\begin{align*}
			\alpha \OPT_{Stoch} + \beta \OPT_{Adv} + \frac{3}{2}\,\psi(p\ones),
		\end{align*}
		where: \vspace{-5pt}
		\begin{itemize}
			\item $\alpha = O(1)^p$ if $\psi$ is positively homogeneous \mnote{Actually need to assume $|Stoch| \ge 4p$, because of restriction of $\bar{\gamma}$}, and $\alpha = O\big(\min\big\{\tfrac{n}{|Stoch|}~,~p^2\big\}\big)^p$ otherwise. In particular, $\alpha$ is always $O(1)^p$ as long as a constant fraction of the times are stochastic \vspace{-3pt}
			\item $\beta = O(p)^p$ if $\psi$ is  separable, and $\beta = O(p)^{2p}$ otherwise.
		\end{itemize}
		%
	\end{thm}
	
	Recall that $\psi$ is said to be \emph{positively homogeneous} if there is $p$ such that every vector $u$ and $\gamma \ge 0$ we have the equality $\psi(\gamma u) = \gamma^p \psi(u)$, and \emph{separable} if there are 1-dimensional functions $\psi_i : \R_+ \rightarrow \R_+$ such that $\psi(u) = \sum_i \psi_i(u_i)$ for all $u$. 
		
	This result recovers and greatly extends the $O(p)^p$-approximation obtained by Gupta et al.~\cite{raviConvexPD} for the pure adversarial model and the special separable case $\psi(u) = \sum_i (\ell_i u_i)^p + \sum_i c_i u_i$ with $\ell, c \ge 0$ (which already has applications in scheduling with speed scaling). It is also known that in the pure adversarial model such $\Omega(p)^p$-approximation is best possible, even with additive errors independent of $n$, even for the special case $\psi(u) = \sum_i u^p$~\cite{molinaro}. 
On the pure stochastic model, not much seems to be known about this problem, so Theorem \ref{thm:main} seems to in particular give the first non-trivial approximation ($O(1)^p$) in this model. 

	Theorem \ref{thm:main} also gives an algorithm for the classic $\ell_p$-norm load-balancing problem~\cite{awerbuch}, which is the special case of \ocp with $\ell_p$-norm objective function $\psi = \|\cdot\|_p$. Here the coordinates can be interpreted as machines, $V_t$ as a job, and each vector $v_t \in V_t$ as a processing option for this job that adds a load of $(v_t)_i$ on the $i$th machine. Then $\|\sum_t v_t\|_p$ is an aggregate measure of the total load incurred over the machines if selecting the processing options $v_t$'s. We remark that $\ell_p$-norm load balancing has been studied since at least the 70's~\cite{chandra,cody}, and that the case $\|\sum_t v_t\|_{\infty}$ corresponds to the standard makespan minimization. Even though the objective function $\psi = \|\cdot\|_p$ does not have non-decreasing gradient, the applying the algorithm from Theorem \ref{thm:main} to $\psi(u) = \|u\|_p^p = \sum_i u_i^p$ gives the following result (details are presented in Appendix \ref{app:loadBal}).
	
	\begin{thm} \label{thm:loadBal}
		There is an algorithm for the $\ell_p$-norm load balancing problem in the \mixed model with expected load at most $O(1)\cdot\OPT_{Stoch} + O(\min\{p, \log m\})\cdot \OPT_{Adv} + O(\min\{p, \log m\}\,m^{1/p})$.
	\end{thm} 
	
	This extends the recent result from~\cite{molinaro} that gave a single algorithm that obtains expected load at most $\approx (1+\e)\, \OPT + \frac{p(m^{1/p} -1)}{\e}$ in the pure stochastic model, and load $\Theta(\min\{p, \log m\}) \cdot \OPT$ in the adversarial model (which is optimal). While the algorithm from~\cite{molinaro} has good guarantees on both pure models, it does not necessarily mean that it maintains them in the \mixed model, as the algorithm from Theorem \ref{thm:loadBal} does. 
	 

	\paragraph{Welfare Maximization with Concave Costs.} We also consider the following general problem with rewards and costs, introduced by Blum et al.~\cite{BGMS11}. Again the convex cost function $\psi$ is known upfront. At each time step, a customer comes with a reward $c_t \in \R$ (say, how much it is willing to pay for a service) and a ``resource consumption'' vector $a_t \in [0,1]^m$ (for this service). Based on the information thus far, the algorithm needs to choose $x_t\in[0,1]$, indicating how much of this request is wants to fulfill. The goal is to maximize the total profit of reward minus cost of resources: $$\sum_{t=1}^n c_t x_t - \psi\bigg(\sum_{t=1}^n a_t x_t\bigg).$$ We refer to this problem as \welf.
	
	Our next contribution is an algorithm for this problem  on the \mixed model, again under the same assumptions of the current literature that requires an additional mild assumption on the cost function~\cite{cvxPDFOCS} (notice that convex polynomials of degree $p$ with non-negative coefficients satisfy the requirements).

	\begin{thm} \label{thm:welfare}
		Consider the \welf problem in the \mixed model. Suppose the cost function $\psi : \R^m_+ \rightarrow \R_+$ satisfies the assumption of Theorem \ref{thm:main}, and that it can be factored as $\psi = \psi_{lin} + \psi_{high}$, where $\psi_{lin}$ is a linear function and $\psi_{high}$ grows at least quadratically: $\psi_{high}(\gamma u) \ge \gamma^2 \,\psi_{high}(u)$ for $\gamma \ge 1$. Then Algorithm~\ref{alg:welfare} has expected profit at least
		\begin{align*}
			\Omega(\max\{\tfrac{|Stoch|}{n}, \tfrac{1}{p}\}) \cdot \OPT_{Stoch} + \Omega(\tfrac{1}{p}) \cdot \OPT_{Adv} - ~O(\psi(p \ones)). 
		\end{align*} 
	\end{thm}	
	
	This result extends the result of Azar et al.~\cite{cvxPDFOCS} that gave a  $\Theta(\frac{1}{p})$-approximation in the pure adversarial model, under the same assumption\footnote{They assume that, in addition to the assumptions from Theorem \ref{thm:main}, $\psi$ is a degree $p$ polynomial, but only the decomposition $\psi = \psi_{lin} + \psi_q$ as in Theorem \ref{thm:welfare} is required.} as in our theorem (following Blum et al.~\cite{BGMS11} and Huang and Kim~\cite{HK15} that gave the same approximation for the special case of separable costs (i.e., $\psi(u) = \sum_i \psi(u_i)$) where each $\psi_i$ is a degree-$p$ polynomial). Even in the pure stochastic model, Theorem \ref{thm:welfare} seems to give the first dimension-independent approximation for this problem. We note, however, that for the pure stochastic model Gupta et al.~\cite{guptaMolRuta} gave a $\Omega(\frac{1}{m})$-approximation without any growth or factorization assumptions on $\psi$.
	


	


	\paragraph{Techniques.} Our primal-dual algorithms are designed using connections with \emph{Online Convex Optimization} (OCO). The informal principle is that in primal-dual algorithms the ``right'' way of updating duals is via a low regret OCO strategy. This principle has been an important component in recent development in online algorithms, not only explaining the ubiquitous exponential dual updates, but also as tool for obtaining the best results for several problems~\cite{guptaMolinaroMOR,AgrawalDevanur,fazel,molinaro}. (Recent results on $k$-Server and related problems~\cite{BCLLM18,guptaKServer,BCLL19,leePureEntropic} also use OCO-based updates, but in the primal.) 
	
	While this connection has been explored both in the pure stochastic~\cite{guptaMolinaroMOR,AgrawalDevanur,molinaro} and pure adversarial~\cite{fazel} settings, the analyzes are quite different in each of the models. In particular this is a main hurdle for analyzing primal-dual algorithms in the \mixed model. So one of our goals is obtaining analyzes that are more \emph{homogeneous}, i.e., can be applied on a per-time step basis, despite the inevitable \emph{disparity} in the loses in each model. For that, one of our main technical ingredients is an OCO algorithm satisfying multiple properties. First, since we want guarantees that are independent on $n$, we cannot work with the standard notion of additive regret, which would typically lead to additive losses $\Omega(\sqrt{n})$. Instead, we work with both additive and a specific type of multiplicative regret. For that, we use both regularization and \emph{shifting} of the OCO functions (the latter was used in the context of $\ell_p$-norms in~\cite{molinaro}). Another issue that precludes the use of off-the-shelf OCO algorithms is that we need a custom control over the ``size'' of the vectors returned by the OCO procedure. Also, ideally we would like these vectors to be increasing: in this case we could majorize (in all coordinates simultaneously) these vectors by the last one, ``factor it out'', and add up the actions over all the adversarial time steps and treat them in one goal. While we are not quite able to guarantee monotonicity while maintaining good enough regret, this is another important factor to balance out in our context that is not present in the standard OCO literature.

		

	\section{Preliminaries}
	
	We will work throughout with ``nice'' convex cost functions.
	
	\begin{ass}[Nice cost functions]
		A function $\psi: \R^m_+ \rightarrow \R_+$ is \emph{nice} if it is convex, differentiable, non-decreasing ($\psi(u') \ge \psi(u)$ for $u' \ge u$), and $\psi(0) = 0$.
	\end{ass}

	
	\paragraph{Fenchel conjugate.} We make use of Fenchel conjugacy throughout the paper, and refer the reader to~\cite{HUL} for more information on the topic.  Given a convex function $\psi : \R^m_+ \rightarrow \R$, its \emph{Fenchel conjugate} is $$\psi^*(y) := \sup_{u \in \R^m} (\ip{u}{y} - \psi(u)).$$ As an example, the conjugate of $\psi(u) = \frac{1}{p} u^p$ is $\psi^*(y) = (1-\frac{1}{p})\,y^{\frac{p}{p-1}}$. For nice functions we have the involution $\psi = \psi^{**}$, a crucial fact that will be used often:
	\begin{align}
		\psi(u) = \sup_{y \in \R^m} \bigg(\ip{y}{u} ~-~ \psi^*(y)\bigg). \label{eq:fenchel}
	\end{align}
	This has the interpretation of representing $\psi$ via its lower bounding linearizations $\ip{y}{\cdot} - \psi^*(y)$. It directly gives the \emph{Fenchel inequality}: 
	\begin{align}
		\ip{y}{u} ~\le~\psi(u) + \psi^*(y)~~~~~\forall u,y. \label{eq:fenchelIneq}
	\end{align}
	
	We collect important properties of the conjugate of a nice function $\psi : \R^m_+ \rightarrow \R_+$ (notice its domain/codomain); see Appendix A.1 of \cite{onlinePDICALP} for proofs.
	
	\begin{lemma} \label{lemma:basicFenchel}
	If $\psi :\R^m_+ \rightarrow \R_+$ is a nice function then: 
	\begin{enumerate}
		\setlength\itemsep{0em}
		\item[a)] The conjugate $\psi^*$ is convex, non-decreasing, non-negative, and has $\psi^*(0) = 0$
		\item[b)] The supremum in \eqref{eq:fenchel} is achieved by $y = \nabla \psi(u)$
		\item[c)] The supremum in \eqref{eq:fenchel} can be taken over only non-negative vectors. 
	\end{enumerate}
	\end{lemma}


	\paragraph{Bounded growth.} We will work throughout with cost functions with growth of order at most $p$, namely $\nabla \psi(\alpha u) \le \alpha^{p-1} \nabla \psi(u)$ for $\alpha \ge 1$. This is satisfied, for example, by all polynomials of degree $p$. We collect some implications of this condition that essentially come from~\cite{cvxPDFOCS}; we provide a proof in Appendix~\ref{app:fenchel}.
		
	\begin{lemma} \label{lemma:growth}
		Consider a nice function $\psi :\R^m_+ \rightarrow \R_+$ with growth of order at most $p > 1$. Then:	
		\begin{enumerate}
			\vspace{-4pt}
			\item[a)]  $\psi(\alpha u) \le \alpha^p\, \psi(u)$ for every vector $u$ and scalar $\alpha \ge 1$.
			\vspace{-6pt}
			\item[b)] $\psi^*(\delta y) \le \delta^{\frac{p}{p-1}}\cdot \psi^*(y)$ for every vector $y$ and scalar $\delta \in (0,1]$
			\vspace{-6pt}
			\item[c)] $\psi^*(\nabla \psi(u)) \le p\cdot \psi(u)$.
		\end{enumerate} 
	\end{lemma}	
	

	\paragraph{Non-decreasing gradients.} We will also make use of the fact that functions with non-decreasing gradients are superadditive, see for example Lemma 9 of~\cite{onlinePDICALP2}.
	
	\begin{lemma} \label{lemma:superadd}
	If $\psi :\R^m_+ \rightarrow \R_+$ is a nice function, then for all $u,v \in \R^m_+$ we have $\psi(u + v) \ge \psi(u) + \psi(v)$.
	\end{lemma}


	\section{Online Convex Programming}

	Here we consider the problem \ocp described in Section \ref{sec:results}, and prove Theorem \ref{thm:main}. The idea for our algorithm is the following. Using the Fenchel conjugacy of \eqref{eq:fenchel}, we can see our problem as the minimax one\footnote{This is just for motivating the algorithm, we do not need to formally invoke a minimax theorem to justify the exchange of $\sup_{y \in \R^m_+}$ and $\min_{(v_t) \in (V_t)}$.}
	\begin{align*}
		\min_{(v_t) \in (V_t)} \psi\left(\sum_t v_t\right) = \sup_{y \in \R^m_+} \min_{(v_t) \in (V_t)}  \left[ \sum_t \ip{y}{v_t} - \sum_t \frac{1}{n} \psi^*(y) \right].
	\end{align*}	
	Given this, the algorithm is by now natural: it is a primal-dual algorithm which computes $v_t$'s (primal) as well as approximations $y_t$'s for the ``right'' $y$ (dual) in an online fashion, using a low regret Online Convex Optimization strategy. More precisely, define the ``Lagrangian''
	\begin{align}
		L(y, v) := \ip{y}{v} - \frac{1}{n}\,\psi^*(y). \label{eq:L}
	\end{align}
	We can think of $L(y,v)$ as the ``fake cost'' at time $t$ if the algorithm plays $v$ obtained by roughly linearizing the objective function $\psi$ with the ``dual'' $y$. The algorithm is then the following. 
	
	\begin{algorithm}[H]
  \caption{Algorithm \alg for \ocp}
  \begin{algorithmic}[1]
  	\For{each time $t$}
		\State (Dual) Compute $\bar{y}_t \in \R^m_+$ feeding  functions $L(\cdot,\bar{v}_1),\ldots,L(\cdot,\bar{v}_{t-1})$ to the OCO algorithm \OCOalg
		\State (Primal) Compute $\bar{v}_t$ as best response for the ``fake cost'': $\bar{v}_t = \argmin_{v_t \in V_t} L(\bar{y}_t,v_t)$
		\EndFor
  \end{algorithmic}
  \label{alg:cvxProg}
\end{algorithm}	

	To get a better intuition for the algorithm, and the type of properties we are looking for in the subroutine \OCOalg, we start with an informal sketch of the analysis. Recalling the notation from Section \ref{sec:results}, for an adversarial time $t \in Adv$ let $v^*_t$ be the choice made by the optimal solution for $Adv$, and for a stochastic time $t \in Stoch$ let $v_t^* := v^*(V_t)$ be the (random) choice made by the optimal selector $v^*$ for $Stoch$. Also let $\vOPT_{Adv} := \sum_{t \in Adv} v^*_t$ and $\vOPT_{Stoch} := \sum_{t \in Stoch} v^*_t$ be the optimal load on the adversarial and stochastic parts. We want to show that the \alg's cost is comparable to that of these loads:  $$\E \psi\bigg(\sum_t \bar{v}\bigg) ~\lesssim~ approx_{Adv} \cdot \psi(\vOPT_{Adv}) ~+~ approx_{Stoch} \cdot \E \psi(\vOPT_{Stoch}).$$
	
	Since the algorithm makes the decisions $\bar{v}_t$ based on the fake costs $L(\bar{y}_t, \bar{v}_t)$, we need to ensure that they reflect the actual cost. For that we need the $\bar{y}_t$'s computed by \OCOalg to be the ``right slope''. More precisely, we want them to have \emph{low regret} with respect to the functions $L(\cdot, \bar{v}_t)$, namely to satisfy
	\begin{align}
		\underbrace{\sum_t L(\bar{y}_t, \bar{v}_t)}_{\textrm{alg's fake cost}} \stackrel{\textrm{regret}}{\gtrsim} \sup_{y \in \R^m_+} \sum_t L(y, \bar{v}_t) = \sup_{y \in \R^m_+} \bigg[\ip{y}{{\textstyle \sum_t \bar{v}_t}} - \psi^*(y)\bigg] = \underbrace{\psi\bigg(\sum_t \bar{v}_t\bigg)}_{\textrm{alg's cost}},  \label{eq:premain}
	\end{align}
	where the loss in the first inequality is what we are informally calling regret. Assuming such regret guarantee is available, using the best response of $\bar{v}_t$ we informally get:
	%
	\begin{align}
		\underbrace{\psi\bigg(\sum_t \bar{v}_t\bigg)}_{\textrm{alg's cost}} \stackrel{\textrm{regret}}{\lesssim} \underbrace{\sum_t L(\bar{y}_t, \bar{v}_t)}_{\textrm{alg's fake cost}} \stackrel{\textrm{best resp}}{\le} \underbrace{\sum_t L(\bar{y}_t, v^*_t)}_{\textrm{$\OPT$'s fake cost}} \stackrel{\star}{\lesssim} \textrm{\OPT's cost}. \label{eq:main}
	\end{align}
	That last inequality needs to be clarified and justified: Why are the algorithm's slopes $\bar{y}_t$ the ``right ones'' for linearizing the unknown $\OPT$? This is more subtle and will depend on the stochastic and adversarial parts of the model. 

 We now make this discussion formal.


	\subsection{Regret part of the analysis} \label{sec:regretPart}
	
	To prove the first inequality in \eqref{eq:main}, we state precisely the type of guarantee needed from the OCO algorithm \OCOalg. We actually work with the slightly more general Lagrangian function
	\begin{align}
		L_\gamma(y, v) := \ip{y}{v} - \gamma\,\psi^*(y)
	\end{align}
	where the multiplier in the last term is parametrized by $\gamma$ instead of being always $\frac{1}{n}$ (though it is instructive to think throughout that $\gamma = \frac{1}{n}$). But we will only consider sequences of multipliers $\gamma_t$ that add up to exactly 1. 
	
	We briefly recall the OCO game~\cite{OCObook}, simplified for our needs. At time $t$, the OCO algorithm needs to produce a vector $y_t \in \R^m_+$ based only on the functions $L_{\gamma_1}(\cdot, v_1), \ldots, L_{\gamma_{t-1}}(\cdot, v_{t-1})$ seen up to the previous time step. It then sees the full reward function $L_{\gamma_t}(\cdot, v_t)$ for this time, and receives reward $L_{\gamma_t}(y_t, v_t)$. The standard goal in the OCO game is obtain reward comparable to using the best vector $y^* \in \R^m_+$ in all time steps (considering $\sum_t \gamma_t = 1$)
	\begin{align*}
		\textrm{reward OCO algo} = \sum_t L_{\gamma_t}(y_t, v_t) \gtrsim  \sup_{y \in \R^m_+} \bigg[\ip{y}{{\textstyle \sum_t v_t}} - \psi^*(y)\bigg] = \psi\bigg(\sum_t v_t\bigg).
	\end{align*}
	But as mentioned before, in our context it will be also crucial to control the ``size'' of the iterates $y_t$ produces, which will play a crucial role in analyzing Algorithm \ref{alg:cvxProg} in the adversarial part of the input. The right notion turns out to be essentially making sure that the ``size'' $\psi^*(y_t)$ is always comparable to the total reward obtained. It will be also important to make sure that the sequence $(y_t)_t$ is almost increasing: if that was the case, we could majorize (in all coordinates simultaneously) these iterates by the last one and then ``factor it out'' so we can add up the cost of all the adversarial time steps (see Lemma \ref{lemma:seqAdv}). In fact, it will be enough if there is a small set of iterates that dominates all the other ones.  
		
	
	We now state the guarantees of our OCO algorithm, but to reduce context-switching we postpone its description and analysis to Section \ref{sec:OCO}. We use $\bigvee_t y_t$ to denote the pointwise maximum of the collection of vectors $\{y_t\}_t$, i.e., the $i$th coordinate of $\bigvee_t y_t$ equals $\max_t (y_t)_i$.
	
	\begin{thm}[Guarantee of \OCOalg] \label{thm:OCO}
		 Assume that $\psi : \R^m_+ \rightarrow \R_+$ satisfies the hypothesis of Theorem~\ref{thm:main}. Consider an instance $(L_{\gamma_t}(\cdot,v_t))_t$ of the OCO game described above, with $v_t \in [0,1]^m$. Assume that the $\gamma_t$'s take value either 0 or $\bar{\gamma} \le \frac{1}{4p}$, and $\sum_t \gamma_t = 1$. 
	 Then the \OCOalg produces a solution $(\bar{y}_t)_t$ with the following properties:
		\begin{enumerate}
			\item (Regret) 
			 $$\sum_t L_{\gamma_t}(\bar{y}_t\,,\,\tfrac{1}{2}v_t) ~\ge~ \psi\bigg(\frac{1}{8} \sum_t v_t\bigg) - \psi(p \ones)$$
			\item (Size control of the iterates) 
				\begin{align*}
					\frac{1}{p} \cdot \max_t \, \psi^*(\bar{y}_t) ~\le~ \sum_{t} \ip{\bar{y}_{t}}{v_{t}} ~+~ \psi(p\ones)
				\end{align*}
				\item (Improved size control in separable case) If in addition the function $\psi$ is separable, we have
				\begin{align*}
					\frac{1}{p} \cdot \psi^*\bigg( \bigvee_t \bar{y}_t\bigg) ~\le~ \sum_{t} \ip{\bar{y}_{t}}{v_{t}}  ~+~ \psi(p\ones) 
				\end{align*}
				
				\item (Almost monotone) There is a set $\mathcal{T} \subseteq [n]$ of size $p$ such that $$\textrm{for every $t$, there is $t' \in \mathcal{T}$ such that ~$\bar{y}_t \le e \cdot \bar{y}_{t'}$}.$$
		\end{enumerate}
	\end{thm}

	Since our primal-dual Algorithm \ref{alg:cvxProg} uses \OCOalg over the sequence of Lagrangian functions $L(\cdot, \bar{v}_t)$, the previous theorem allows us to relate the algorithms real and fake costs, giving the first inequality in \eqref{eq:main}.
	
	\begin{cor} \label{cor:cor}
		The cost of Algorithm \ref{alg:cvxProg} satisfies
		\begin{align*}
			\psi\bigg(\frac{1}{8} \sum_t \bar{v}_t\bigg) ~\le~ \sum_t L(\bar{y}_t\,,\bar{v}_t) - \frac{1}{2p} \cdot \max_t \psi^*(\bar{y}_t) + \frac{3}{2} \psi(p \ones).
		\end{align*}
		Moreover, if the function $\psi$ is separable, the second term in the right-hand side can be replaced by $\frac{1}{2p} \psi^*(\bigvee_t \bar{y}_t)$.
	\end{cor}
	
	\begin{proof}
		Notice that $L(\bar{y}_t\,,\frac{1}{2} \bar{v}_t) + \frac{1}{2} \ip{\bar{y}_t}{\bar{v}_t} = L(\bar{y}_t\,,\bar{v}_t)$. The result for non-separable functions then follows by adding Item 1 and half of Item 2 of Theorem \ref{thm:OCO} and rearranging the terms. The result for separable functions simply uses Item 3 instead of Item 2. 
	\end{proof}
	
	Moreover, since by the best response property of $\bar{v}_t$ \alg's fake cost $L(\bar{y}_t,\bar{v}_t)$ is at most $\OPT$'s fake cost $L(\bar{y}_t,v^*_t)$, we now upper bound the latter, formalizing the last inequality in \eqref{eq:main}. 

	
%
%


	\subsection{Relating $\OPT$'s fake and real cost}

	We consider the stochastic and adversarial time steps separately, starting with the stochastic part.

	
	
	\begin{lemma}[Stochastic part] \label{lemma:stoch}
		We have
		\begin{align*}
			\E \sum_{t \in Stoch} L(\bar{y}_t, v^*_t) ~\le~ \frac{1}{\beta}\,\psi\big(\beta\,\E \vOPT_{Stoch}\big), 
		\end{align*}
		where $\beta = \frac{n}{|Stoch|}$.
	\end{lemma}
	
	\begin{proof}
		Consider a stochastic time step $t$. Notice that both $\bar{y}_t$ and $v^*_t$ are random variables. Let $\F_{t-1}$ denote the $\sigma$-algebra generated by the history up to time $t-1$, i.e., by the sequence $(V_{t'})_{t' \in Stoch, t' \le t-1}$. Since $\bar{y}_t$ only depends on this history, conditioning on $\F_{t-1}$ fixes $\bar{y}_t$, but does not affect the distribution of $v^*_t$. In addition, notice that $\E v^*_t = \frac{1}{|Stoch|}\, \E \vOPT_{Stoch}$.	Therefore,
		\begin{align}
			\E\Big[ L(\bar{y}_t, v^*_t) ~\Big|~ \F_{t-1} \Big] &= \ip{\bar{y}_t\,}{\,\E[v^*_t \mid \F_{t-1}]} - \frac{1}{n} \psi^*(\bar{y}_t) \notag\\
	&= \ip{\bar{y}_t}{\E v^*_t} - \frac{1}{n} \psi^*(\bar{y}_t)\notag\\
	&= \frac{1}{n} \left(\ip{\bar{y}_t\,}{\beta\, \E \vOPT_{Stoch}} - \psi^*(\bar{y}_t) \right)\notag\\
	&\le \frac{1}{n}\, \psi\big(\beta\, \E \vOPT_{Stoch}\big),	\label{eq:stoch}
		\end{align}
		where the last inequality holds by the Fenchel conjugacy \eqref{eq:fenchel}. Adding over all stochastic time steps concludes the proof. 
	\end{proof}

	

	Now we consider the adversarial time steps. Unlike in the previous lemma, we cannot analyze each such time individually, since each by itself is not ``representative'' of $\vOPT_{Adv}$. It is here that we need an almost monotonicity property to ``factor the $\bar{y}_t$'s out''.
	
	\begin{lemma}[Adversarial part] \label{lemma:seqAdv}
	For all $\alpha \ge 1$
	\begin{align}
		\sum_{t \in Adv} L(\bar{y}_t, v^*_t) ~\le~ e \cdot \psi\big(\alpha \vOPT_{Adv}\big) \,+\, \frac{e p}{\alpha} \,\max_t\, \psi^*(\bar{y}_t) \label{eq:seqAdv1}
	\end{align}
	and
	\begin{align}
		\sum_{t \in Adv} L(\bar{y}_t, v^*_t) ~\le~ \psi\big(\alpha \vOPT_{Adv}\big) \,+\, \frac{1}{\alpha} \,\psi^*\bigg(\bigvee_t \bar{y}_t\bigg). \label{eq:seqAdv2}
	\end{align}
	\end{lemma}
	
	\begin{proof}[Proof of \eqref{eq:seqAdv1}]
		Consider the set $\mathcal{T} = \{t_1,t_2,\ldots,t_p\}$ in Item 4 of the guarantee of \OCOalg from Theorem~\ref{thm:OCO}, and partition the adversarial time steps $Adv$ into sets $Adv_1,Adv_2,\ldots,Adv_p$ so that for all $t \in Adv_i$ we have the domination $\bar{y}_t \le e \cdot \bar{y}_{t_i}$.

	Since $\psi^*$ is non-negative, $L(\bar{y}_t,v^*_t)$, is at most $\ip{\bar{y}_t}{v^*_t}$, which equals $\ip{\bar{y}_t/\alpha}{\alpha v^*_t}$. Then using the domination above,
		\begin{align}
			\sum_{t \in Adv} L(\bar{y}_t, v^*_t) &\le \sum_{i \in [p]} \sum_{t \in Adv_i} \ip{\bar{y}_t/\alpha}{\, \alpha v^*_t} \le e \cdot \sum_{i \in [p]} \sum_{t \in Adv_i} \ip{\blue{\bar{y}_{t_i}}/\alpha\,}{\,\alpha v^*_t} = e \cdot \sum_{i \in [p]} \ip{\bar{y}_{t_i}/\alpha\,}{\, \alpha \sum_{t \in Adv_i} v^*_t}. \label{eq:seqAdvProof}
		\end{align}
		Applying Fenchel's inequality \eqref{eq:fenchelIneq} to each term of the right-hand side we obtain
		\begin{align}
			RHS &\le e \cdot \sum_{i \in [p]} \bigg[\psi\bigg(\alpha \sum_{t \in Adv_i} v^*_t\bigg) ~+~ \psi^*\bigg(\frac{1}{\alpha}\, \bar{y}_{t_i}\bigg) \bigg] \notag\\
			&\le e \cdot \psi\bigg(\alpha \sum_{t \in Adv} v^*_t\bigg) + ep \cdot \max_t \psi^*\bigg(\frac{1}{\alpha}\, \bar{y}_t\bigg) \label{eq:seqAdvProof2}\\
			&\le e \cdot \psi\bigg(\alpha \sum_{t \in Adv} v^*_t\bigg) + \frac{ep}{\alpha} \cdot \max_t \psi^*( \bar{y}_t),	\notag
		\end{align}
	where the second inequality uses the superadditivity of $\psi$ (Lemma \ref{lemma:superadd}), and the last inequality uses Lemma~\ref{lemma:growth}.(b). This concludes the proof of inequality \eqref{eq:seqAdv1}.
	\end{proof}
	
	\begin{proof}[Proof of \eqref{eq:seqAdv2}]
		The proof is a simplified version of the previous argument using the coarser domination $\bar{y}_t \le \bigvee_t \bar{y}_t$: inequality \eqref{eq:seqAdvProof} now becomes
		\begin{align*}
			\sum_{t \in Adv} L(\bar{y}_t, v^*_t) \le \ip{\bigvee_t \bar{y}_t}{\sum_{t \in Adv} v^*_t},
		\end{align*}
		and the argument from inequality \eqref{eq:seqAdvProof2} gives that the right-hand side is at most $\psi(\alpha \sum_{t \in Adv} v^*_t) + \frac{1}{\alpha} \psi^*(\bigvee_t \bar{y}_t).$ This concludes the proof. 
	\end{proof}
		

	\subsection{Concluding the proof of Theorem \ref{thm:main}}
	
	Now we just need to plug in the pieces above to conclude the proof of Theorem \ref{thm:main}. Since we get different guarantees depending on whether $\psi$ is separable/homogeneous or not, we need to consider these cases separately. 
	
	\subsubsection{Non-separable, non-homogeneous case} \label{sec:nonsep}
	
	Start from Corollary \ref{cor:cor}, use the fact that $L(\bar{y}_t\,,\,\bar{v}_t) \le L(\bar{y}_t\,,\,v^*_t)$ (since $\bar{v}_t$ is best response), and upper bound the costs $L(\bar{y}_t\,,\,v^*_t)$ using Lemmas \ref{lemma:stoch} and the first part of \ref{lemma:seqAdv} (with $\alpha = 2 e p^2$, so that the terms $\max_t \psi^*(\bar{y}_t)$ cancel out). This gives 
	\begin{align*}
		\E\, \psi\bigg(\frac{1}{8} \sum_t v_t\bigg) \le e \cdot \psi \big(2ep^2\,\vOPT_{Adv}\big) + \psi\big(\beta\,\E\vOPT_{Stoch}\big) + \frac{3}{2} \psi(p \ones).  
	\end{align*}
	Applying the growth condition $\psi(\gamma z) \le \gamma^p \psi(z)$ (Lemma \ref{lemma:growth}.(a)) to pull out the multipliers and using the fact $\psi$ is convex plus Jensen's inequality to pull out the expectation further gives that
	\begin{align}
		\E\, \psi\bigg(\sum_t v_t\bigg) \le   O(p)^{2p}\, \psi \big(\vOPT_{Adv}\big) ~+~ O\big(\tfrac{n}{|Stoch|}\big)^p\cdot \E \psi\big(\vOPT_{Stoch}\big) ~+~ \frac{3}{2} \psi(p \ones). \label{eq:wrapMain0}
	\end{align}
	This essentially gives Theorem \ref{thm:main} in the non-separable case, other than we want to replace the term $O\big(\frac{n}{|Stoch|}\big)^p$ by $\min\{O(\frac{n}{|Stoch|})^p\,,\,O(p)^{2p}\}$. For that we can simply treat \emph{all} time steps as adversarial, in which case the previous bound gives (in each scenario)
	\begin{align}
		\psi\bigg(\sum_t v_t\bigg) \le O(p)^{2p} \psi \big(\vOPT_{Adv} + \vOPT_{Stoch}\big)~+~ \frac{3}{2}\psi(p \ones).	\label{eq:wrapMain1}
	\end{align}
	In addition, the convexity and growth upper bound on $\psi$ gives the following.
	
	\begin{lemma}
		For every $u,v \in \R^m_+$, we have $\psi(u + v) \le 2^{p-1} \psi(u) + 2^{p-1} \psi(v)$.
	\end{lemma}
	
	\begin{proof}
		Using convexity and the growth upper bound we have $$\psi(u+v) ~=~ \psi\bigg(\frac{1}{2}(2u + 2v)\bigg) ~\le~ \frac{1}{2} \psi(2u) + \frac{1}{2} \psi(2v) ~\le~ 2^{p-1} \psi(u) + 2^{p-1} \psi(v).$$
	\end{proof}
	
	Applying this to \eqref{eq:wrapMain1} and taking expectations we get 
	\begin{align*}
		\E \psi\bigg(\sum_t v_t\bigg) \le O(p)^{2p} \psi \big(\vOPT_{Adv}) ~+~ O(p)^{2p} \,\E \psi\big(\vOPT_{Stoch}\big)~+~ \frac{3}{2} \psi(p \ones).
	\end{align*}
	Combining this bound with the bound from \eqref{eq:wrapMain1} we conclude the proof of Theorem \ref{thm:main} in the non-separable case. 


	\subsubsection{Separable, non-homogeneous case} \label{sec:separable}

	As before, we start from the \emph{separable} version of Corollary \ref{cor:cor}, use the fact that $L(\bar{y}_t\,,\,\bar{v}_t) \le L(\bar{y}_t\,,\,v^*_t)$ (since $\bar{v}_t$ is best response), and upper bound the costs $L(\bar{y}_t\,,\,v^*_t)$ using Lemmas \ref{lemma:stoch} and the \emph{second} part of \ref{lemma:seqAdv}, \emph{but now we use $\alpha = 2p$} (so the terms $\psi^*(\bigvee_t \bar{y}_t)$ cancel out). This gives 
	\begin{align*}
		\E\, \psi\bigg(\frac{1}{8} \sum_t v_t\bigg) \le \psi \big(2p \,\vOPT_{Adv}\big) + \psi\big(\beta\,\E\vOPT_{Stoch}\big) + \frac{3}{2} \psi(p \ones).  
	\end{align*}
 Again applying the growth condition $\psi(\gamma z) \le \gamma^p \psi(z)$ (Lemma \ref{lemma:growth}.(a)) to pull out the multipliers and using the fact $\psi$ is convex plus Jensen's inequality to pull out the expectation gives that 
	\begin{align*}
		\E\, \psi\bigg(\sum_t v_t\bigg) \le O(p)^p\,\psi \big(\vOPT_{Adv}\big) ~+~ O\big(\tfrac{n}{|Stoch|}\big)^p\cdot\E\,\psi\big(\vOPT_{Stoch}\big) + \frac{3}{2} \psi(p \ones).   
	\end{align*}
	The term $O\big(\frac{n}{|Stoch|}\big)^p$ can again be replaced by $\min\{O(\frac{n}{|Stoch|})^p\,,\,O(p)^{2p}\}$ via the same argument as in the previous case. This proves Theorem \ref{thm:main} in the separable case.


	\subsubsection{Homogeneous case} \label{sec:homo}
	
	When $\psi$ is homogeneous we want to get guarantees for the non-separable/separable cases as above but with the approximation factor in the stochastic part being simply $O(1)^p$ instead of $\min\{O(\frac{n}{|Stoch|})^p\,,\,O(p)^{2p}\}$. Tracking down the factor $\frac{n}{|Stoch|}$ we see that it appeared in Lemma \ref{lemma:stoch} because of the term $\frac{1}{n} \psi^*(\bar{y}_t)$ in our Lagrangian $L(\bar{y}_t, \bar{v}_t)$. Indeed, if we had used the modified Lagrangian $\wtilde{L}(\bar{y}_t, \bar{v}_t) = \ip{\bar{y}_t}{\bar{v}_t} - \blue{\gamma_t} \psi^*(\bar{y}_t)$ where $\gamma_t = \frac{1}{|Stoch|}$ for the stochastic times and $\gamma_t = 0$ otherwise, the argument of Lemma \ref{lemma:stoch} would give $\E \sum_{t \in Stoch} \wtilde{L}(\bar{y}_t, \bar{v}_t) \,\le\, \E \psi(\vOPT_{Stoch})$, without the factor $\frac{n}{|Stoch|}$.
	
	The issue is that since the algorithm does not known which times are stochastic it cannot compute such $\gamma_t$'s and feed the modified Lagrangians $\wtilde{L}(\bar{y}_t, \bar{v}_t)$ to OCO algorithm \OCOalg. However, when $\psi$ is homogeneous \emph{the choices $\bar{v}_t$ made by Algorithm \ref{alg:cvxProg} are the same whether it feeds the regular or the modified Lagrangians to \OCOalg}. At a high-level this is because:\vspace{-3pt}
	
	\begin{enumerate}
		\item The iterates returned by \OCOalg are evaluations of the gradient of $\psi$ at different points \vspace{-3pt}
		
		\item Using the $\gamma_t$'s instead of $\frac{1}{n}$ only introduces a time-dependent scaling $\alpha_t$ of the point where $\nabla \psi$ is evaluated, i.e. $\nabla \psi(\alpha_t u)$ vs. $\nabla \psi(u)$ \vspace{-3pt}
		
		\item Since $\psi$ is homogeneous, the gradient $\nabla \psi$ is also homogeneous (of power $p-1$). Thus, we can pull out the $\alpha_t$, i.e., $\alpha_t^{(p-1)} \cdot \nabla \psi(u)$ vs. $\nabla \psi(u)$. \vspace{-3pt}
		
		\item Thus, changing the Lagrangian only rescales the iterates produced by \OCOalg, but this does not change the best response $\bar{v}_t$.
	\end{enumerate}
	
	Since making points 1 and 2 above formal involves getting inside algorithm \OCOalg, we postpone the details to Appendix \ref{app:homo}.
	

	\section{OCO algorithm \OCOalg} \label{sec:OCO}
	
	We finally describe the algorithm Shifted and Scaled Follow the Regularized Leader (\OCOalg), and prove its guarantees from Theorem \ref{thm:OCO}. While this section is self-contained and does not require background on Online Convex Optimization, it draws heavy inspiration from it, for which we refer the reader to~\cite{OCObook}.
	
	Recall that we want to compute $\bar{y}_t \in \R^m_+$ based only on the functions $L_{\gamma_1}(\cdot, v_1), \ldots, L_{\gamma_{t-1}}(\cdot, v_{t-1})$ seen up to the previous time step. A main goal is to maximize $\sum_t L_{\gamma_t}(\bar{y}_t,v_t)$, obtaining guarantees $\sum_t L_{\gamma_t}(\bar{y}_t,v_t) \gtrsim \sup_{y \in \R^m_+} \sum_t L_{\gamma_t}(y, v_t)$ (plus other properties). 

		To absorb some of the losses, define the scaled version of these functions 		
	 $$\wtilde{L}_t(y) := \ip{y}{v_t} - 4 \gamma_t \psi^*(y) ~=~ 4 \cdot L_{\gamma_t}(y\,,\, \tfrac{1}{4} v_t).$$ In addition, define a ``fake'' gain function at time 0: $$\wtilde{L}_0(y) := \ip{y}{4 p \ones} - 4 \psi^*(y).$$ Then \OCOalg is FTRL over the functions $\wtilde{L}_t$ with an additional regularizer $\bar{\gamma} \psi^*(y)$ (recall we assumed $\gamma_t$ is either 0 or $\bar{\gamma}$):

	\begin{algorithm}[H]
  \caption{\OCOalg}
  \begin{algorithmic}[1]
  	\State For time $t$, play
  	\vspace{-4pt}
		\begin{align}
\bar{y}_t &= \argmax_y \bigg(\wtilde{L}_0(y) + \ldots + \wtilde{L}_{t-1}(y) ~-~ \bar{\gamma} \psi^*(y)\bigg). \label{eq:mFTRLn} 
		\end{align}
	\end{algorithmic}
  \label{alg:oco}
\end{algorithm}
	
	Let $v_{1:t} := v_1 + \ldots v_t$, and define $\gamma_{1:t}$ analogously. Using the fact that the gradient is the maximizer in the Fenchel conjugate (Lemma \ref{lemma:basicFenchel}.(b)) we can also see the iterate $\bar{y}_t$ as a gradient of $\psi$: 
	\begin{align}
		\bar{y}_t &= \argmax_y \bigg(\ipbig{y}{4p\ones + v_{1:t-1}} - 4\,\bigg(1 + \gamma_{1:t-1} + \bar{\gamma}  \bigg)\,\psi^*(y)\bigg) \notag\\
		&= \argmax_y \bigg(\ipbig{y}{\frac{4p\ones + v_{1:t-1}}{4(1 + \gamma_{1:t-1} + \bar{\gamma})}} ~-~ \psi^*(y)\bigg) \notag\\
		& = \nabla \psi\bigg(\frac{4p\ones + v_{1:t-1}}{4(1 + \gamma_{1:t-1} + \bar{\gamma})}\bigg). \label{eq:OCOn}
	\end{align}
	
	The fake gain $\wtilde{L}_0$ plays a crucial role in the stability of the algorithm, ensuring that (multiplicatively) $\bar{y}_{t+1} \approx \bar{y}_t$: Its term $4\psi^*(y)$ is the regularizer of the algorithm. In addition, its term $\ip{y}{4p\ones}$ ``shifts everything away from the origin''. To see why this is important for \emph{multiplicative} stability, from \eqref{eq:OCOn} we see that the iterates $\bar{y}_{t+1}$ and $\bar{y}_t$ are evaluations of the gradients of $\psi$ at ``nearby'' points (recall $v_t \in [0,1]^m$). However, for points very close to the origin, let alone the origin itself, these gradients can be extremely different multiplicatively, even in simple 1-dim functions like $\psi(u) = u^p$, where  $\lim_{\e \rightarrow 0} \frac{\psi'(1+\e)}{\psi'(\e)} = \infty$. Notice that at least for the function $\psi(u) = u^p$ this does not happen away from the origin: e.g., $\frac{\psi(p + 1)}{\psi(p)} = (1+\frac{1}{p})^p \le e$. This idea of shifting away from the origin is inspired in a similar strategy used in~\cite{molinaro} in the special case of controlling the gradients of the $\ell_p$ norms. 
	
	Now we analyze \OCOalg, proving Theorem \ref{thm:OCO}.
	

	\paragraph{Item 1: Regret.} To bound the regret of \OCOalg we use a Be-the-Leader/Follow-the-Leader argument. It will be convenient to work with the Follow-the-Leader iterates
	\begin{align}
		\wtilde{y}_t ~:=~ \argmax_y \bigg(\wtilde{L}_0(y) + \ldots + \wtilde{L}_{t-1}(y)\bigg) ~=~ \nabla \psi\bigg(\frac{4p\ones + v_{1:t-1}}{4(1 + \gamma_{1:t-1})}\bigg),  \label{eq:FTL}
	\end{align}
	that is, we simply omit the additional regularized $\bar{\gamma} \psi^*(y)$ from \eqref{eq:mFTRLn}.
	
	We start with the Be-the-Leader argument, which states that playing the next iterate $\wtilde{y}_{t+1}$ at time $t$ (over all $t$) gives superoptimal gains with respect to the $\wtilde{L}_t$'s. We include a proof for completeness.
	
	\begin{lemma}[Be-the-Leader] \label{lemma:BTL}
		For all times $t$, $\sum_{t' = 0}^t \wtilde{L}_{t'}(\wtilde{y}_{t'+1}) \ge \max_y \sum_{t'=0}^t \wtilde{L}_{t'}(y)$.
	\end{lemma}
	
	\begin{proof}
		By induction on $t$: assuming this holds for $t$, using the optimality of $\wtilde{y}_{t+2}$ we get
		\begin{align*}
			 \max_y \sum_{t'=1}^{t + 1} \wtilde{L}_{t'}(y) = \sum_{t' = 0}^{t + 1}  \wtilde{L}_{t'}(\wtilde{y}_{t+2}) \stackrel{\textrm{induction}}{\le} \sum_{t'=0}^{t} \wtilde{L}_{t'}(\wtilde{y}_{t'+1}) ~+~\wtilde{L}_{t+1}(\wtilde{y}_{t+2}), 
		\end{align*}
		concluding the induction. 
	\end{proof}
	
	Now we show that the iterate $\bar{y}_t$ that we construct for time $t$ is almost the same as $\wtilde{y}_{t+1}$.
	
	\begin{lemma} \label{lemma:stab}
		For all $t$ we have $\bar{y}_t \le \wtilde{y}_{t+1} \le 2 \,\bar{y}_t$.
	\end{lemma}

	\begin{proof}
		Let $\bar{w}_t$ be the argument in the gradient in \eqref{eq:OCOn} so that $\bar{y}_t = \nabla \psi(\bar{w}_t)$, and similarly let $\wtilde{w}_t$ be the argument in the gradient in \eqref{eq:FTL}. The ratio of the $j$th coordinate of the vectors $\wtilde{w}_{t+1}$ and $\bar{w}_t$ is 
		\begin{align}
		\frac{(\wtilde{w}_{t+1})_j}{(\bar{w}_t)_j} ~=~ \bigg[\frac{4p + (v_1)_j + \ldots + (v_t)_j}{4p + (v_1)_j + \ldots + (v_{t-1})_j}\bigg] \cdot \bigg[\frac{1 + \gamma_{1:t-1} + \bar{\gamma}}{1 + \gamma_{1:t}} \bigg].  \label{eq:bra}
		\end{align}
	Since $v_t \in [0,1]^m$, the ratio in the first bracket is at least 1; for the same reason, it is also at most $(\sqrt{2})^{1/p}$:
		\begin{align*}
		\frac{4p + (v_1)_j + \ldots + (v_t)_j}{4p + (v_1)_j + \ldots + (v_{t-1})_j} ~=~ 1 + \frac{(v_t)_j}{4p + (v_1)_j + \ldots + (v_{t-1})_j} ~\le~ 1 + \frac{1}{4p} ~\le~ (\sqrt{2})^{1/p},
		\end{align*}
		where the last inequality uses the bound $(1+\frac{x}{p})^p \le e^{x}$ valid for all $x$. Similarly, the second bracket of \eqref{eq:bra} is also are least 1 (since $\gamma_t \le \bar{\gamma}$) and at most $(\sqrt{2})^{1/p}$:
		\begin{align*}
		\frac{1 + \gamma_{1:t-1} + \bar{\gamma}}{1 + \gamma_{1:t}} \le 1 + \bar{\gamma} \le (\sqrt{2})^{1/p},
		\end{align*}		
		where the last inequality uses the assumption that $\bar{\gamma}$ is at most $\frac{1}{4p}$.
		
		Together these give that $\bar{w}_t \le \wtilde{w}_{t+1} \le 2^{1/p} \bar{w}_t$. Then our assumption that the gradient $\nabla \psi$ is monotone guarantees that $$\bar{y}_t = \nabla \psi(\bar{w}_t) \le \nabla \psi(\wtilde{w}_{t+1}) = \wtilde{y}_{t+1},$$ and additionally using the growth condition we get $$\wtilde{y}_{t+1} \le \nabla \psi(2^{1/p} \bar{w}_t) \le 2 \nabla \psi(\bar{w}_t) = 2 \bar{y}_t.$$ This concludes the proof. 
	\end{proof}
	
	Using these two lemmas we can finalize the proof of Item 1 of Theorem \ref{thm:OCO}. Since $\psi^*$ is non-decreasing, the previous lemma gives
	\begin{align*}
		L_{\gamma_t}(\bar{y}_t, \tfrac{1}{2} v_t) ~=~ \frac{1}{2} \ip{\bar{y}_t}{v_t} - \gamma_t\,\psi^*(\bar{y}_t) ~\ge~ \frac{1}{4} \ip{\wtilde{y}_{t+1}}{v_t} - \gamma_t \psi^*(\wtilde{y}_{t+1}) ~=~ \frac{1}{4} \wtilde{L}_t(\wtilde{y}_{t+1}).
	\end{align*}
	Adding over all times $t' \le t$ using Lemma \ref{lemma:BTL} we get
	\begin{align}
		\sum_{t' = 1}^{t} L_{\gamma_t}(\bar{y}_{t'}, \tfrac{1}{2} v_{t'}) ~\ge~ \frac{1}{4} \sum_{t' = 1}^{t} \wtilde{L}_{t'}(\wtilde{y}_{t'+1}) &~\ge~ \frac{1}{4} \max_y \sum_{t' = 0}^{t} \wtilde{L}_{t'}(y) ~-~ \frac{1}{4} \wtilde{L}_0(\wtilde{y}_1). \label{eq:regretN1}
	\end{align}
	Expanding the first term in the right-hand side we see
	\begin{align*}
		\max_y \sum_{t' = 0}^{t} \wtilde{L}_{t'}(y) ~&=~ \max_y \bigg(\ip{y}{4p\ones + v_1 + \ldots + v_t} -  4\bigg(1 + \gamma_{1:t} \bigg) \psi^*(y)\bigg)\\
		&=~ 4\bigg(1 + \gamma_{1:t}\bigg) \cdot \max_y \bigg(\ipbig{y}{\frac{4p\ones + v_1 + \ldots + v_t}{4(1 + \gamma_{1:t})}} -  \psi^*(y)\bigg)\\
		&=~ 4\bigg(1 + \gamma_{1:t}\bigg) \cdot \psi\bigg(\frac{4p\ones + v_1 + \ldots + v_t}{4(1 + \gamma_{1:t})}\bigg).
	\end{align*}	
	Similarly, the optimality of $\wtilde{y}_1$ gives
	\begin{align*}
	\wtilde{L}_0(\wtilde{y}_1) = \sup_y \wtilde{L}_0(y) = \sup_y \bigg(\ip{y}{4p\ones} - 4 \psi^*(y) \bigg) = 4 \psi(p \ones).
	\end{align*}	
	Plugging these bounds on \eqref{eq:regretN1} gives that for all $t$
	\begin{align}
	\sum_{t' = 1}^{t} L_{\gamma_t}(\bar{y}_{t'}, \tfrac{1}{2} v_{t'}) ~\ge~ \psi\bigg(\frac{4p\ones + v_1 + \ldots + v_t}{4(1 + \gamma_{1:t})}\bigg) ~-~  \psi(p \ones). \label{eq:regretN2}
	\end{align}	
	Taking $t=n$ and recalling $\sum_t \gamma_t = 1$ proves Item 1 of Theorem \ref{thm:OCO}.
	

	\paragraph{Item 2: Size control of the iterates.} Using the expression for $\bar{y}_t$ given by equation \eqref{eq:OCOn} we have
	\begin{align*}
		\psi^*(\bar{y}_t) ~=~ \psi^*\bigg(\nabla \psi\bigg(\frac{4 p\ones + v_1 + \ldots v_t}{4 (1 + \gamma_{1:t-1} + \bar{\gamma})} \bigg)  \bigg) \le \psi^*\bigg(\nabla \psi\bigg(\frac{4 p\ones + v_1 + \ldots v_t}{4 (1 + \blue{\gamma_{1:t}})} \bigg)  \bigg) ~\le~ p \cdot \psi\bigg(\frac{4p\ones + v_1 + \ldots v_t}{4 (1 + \gamma_{1:t})} \bigg),
	\end{align*}
	where the first inequality uses $\gamma_t \le \bar{\gamma}$ and the monotonicity of both $\psi^*$ and $\nabla \psi$, and the last inequality uses Lemma~\ref{lemma:growth}.(c). Plugging this into inequality \eqref{eq:regretN2} and noticing that $L_{\gamma_t}(\bar{y}_t, \tfrac{1}{2}v_t) \le \ip{\bar{y}_t}{v_t}$, we get
	\begin{align}
		\frac{1}{p} \psi^*(\bar{y}_t) ~\le~ \sum_{t' = 1}^{t} \ip{\bar{y}_{t'}}{v_{t'}} ~+~ \psi(p\ones) ~\le~   \sum_t \ip{\bar{y}_t}{v_t} ~+~ \psi(p\ones),  \label{eq:regretN3}
	\end{align}
	proving the desired result. 


	\paragraph{Item 3: Improved size control in separable case.} Now $\psi : \R^m_+ \rightarrow \R_+$ is a separable function, so we can write it as $\psi(x) = \sum_i \psi_i(x_i)$ for 1-dimensional functions $\psi_i : \R_+ \rightarrow \R_+$. It follows directly from the definition of Fenchel conjugate that $\psi^*$ is also separable: $\psi^*(y) = \sum_i \psi^*_i(y_i)$, where $\psi_i^*$ is the Fenchel dual of $\psi_i$.  
	
	The main point is that in this separable case the algorithm \OCOalg acts \emph{independently} on each of the coordinates, namely the $i$th coordinate of the iterate $\bar{y}_t$ only depends on the $i$th coordinate of the $v_t$'s; for example, from \eqref{eq:OCOn} we see that
	\begin{align*}
		(\bar{y}_t)_i ~=~ \psi'_i\bigg(\frac{4p + (v_1)_i + \ldots + (v_{t-1})_i}{4(1 + \gamma_{1:t-1} + \bar{\gamma})} \bigg).
	\end{align*}
	Thus, we can equivalently think that we are running copies of \OCOalg on $m$ 1-dimensional problems, one for each coordinate. Moreover, since each $\psi_i$ is a coordinate restriction of $\psi$, it is easy to see that $\psi_i$ still satisfies all the properties needed for the analysis of Items 1 and 2 above (e.g., has non-decreasing derivatives, has the bounded growth condition $\psi'_i(\alpha x_i) \le \alpha^p \psi'_i(x_i)$, etc.). In particular, inequality \eqref{eq:regretN3} holds for $\psi_i$, giving
	\begin{align*}
		\sum_t (\bar{y}_t)_i\cdot(v_t)_i ~+~ \psi_i(p) ~\ge~ \frac{1}{p} \max_t \psi^*_i((\bar{y}_t)_i) ~=~ \frac{1}{p} \psi^*_i\big(\max_t (\bar{y}_t)_i\big),
	\end{align*}
	where the last equation uses the fact that $\psi^*$ is non-decreasing (Lemma \ref{lemma:basicFenchel}.(a)). Since the $i$th coordinate of $\bigvee_t \bar{y}_t$ is precisely $\max_t (\bar{y}_t)_i$, adding the displayed inequality over all coordinates $i$ gives
	\begin{align*}
		\sum_t \ip{\bar{y}_t}{v_t} ~+~ \psi(p \ones) ~\ge~ \frac{1}{p} \sum_i \psi^*_i\bigg(\bigg(\bigvee_t \bar{y}_t\bigg)_i\bigg) ~=~ \frac{1}{p} \psi^*\bigg(\bigvee_t \bar{y}_t\bigg),
	\end{align*}	
	yielding the desired result. 	
	

	\paragraph{Item 4: Almost monotone.} For $i=1,\ldots,p$, let $t_i$ be a time so that $\gamma_{1:t_i} \in [2^{i/p} - 1,~ 2^{i/p} - 1 + \bar{\gamma}]$. Then set $\mathcal{T}$ to be the set of these times $t_i$'s. 
	
	To show that for every time $t$, there is a $t_i$ in $\mathcal{T}$ such that $\bar{y}_t \le e  \bar{y}_{t_i}$, we actually show something stronger: take a time $t$ in the interval $(t_{i-1}, t_i]$; we show that $\bar{y}_t \le e  \bar{y}_{t_i}$. For that, let again $\bar{w}_t$ be the argument in the gradient in \eqref{eq:OCOn} so that $\bar{y}_t = \nabla \psi(\bar{w}_t)$. The ratio of the $j$th coordinate of the vectors $\bar{w}_t$ and $\bar{w}_{t_i}$ is
		\begin{align}
		\frac{(\bar{w}_t)_j}{(\bar{w}_{t_i})_j} ~=~ \bigg[\frac{4p + (v_1)_j + \ldots + (v_t)_j}{4p + (v_1)_j + \ldots + (v_{t_i})_j}\bigg] \cdot \bigg[\frac{1 + \gamma_{1:t_i-1} + \bar{\gamma}}{1 + \gamma_{1:t-1} + \bar{\gamma}} \bigg].
		\end{align}
		Since the $v_t$'s are non-negative and $t \le t_i$, the first bracket in the right-hand side is at most 1. In addition, since $t > t_{i-1}$, the second bracket is at most 
		\begin{align*}
		\frac{1 + \gamma_{1:t_i-1} + \bar{\gamma}}{1 + \gamma_{1:t_{i-1}-1} + \bar{\gamma}} \le \frac{1 + \gamma_{1:t_i-1}}{1 + \gamma_{1:t_{i-1}-1}} \le  \frac{2^{i/p} + \bar{\gamma}}{2^{(i-1)/p}} \le 2^{1/p} (1 +  \bar{\gamma}) \le e^{1/p},
		\end{align*}
		where the last inequality uses $(1+x) \le e^x$ and that $\bar{\gamma} \le \frac{1}{4p}$. Therefore, we have that $\bar{w}_t \le e^{1/p} \bar{w}_{t_i}$, and using the growth condition of $\nabla \psi$ we have $$\bar{y}_t = \nabla \psi(\bar{w}_i) \le \nabla \psi(e^{1/p} \bar{w}_{t_i}) \le e\,\nabla \psi(\bar{w}_{t_i}) = e \bar{y}_{t_i}.$$ This concludes the proof of Theorem \ref{thm:OCO}. 
			

	\section{Welfare Maximization with Convex Costs}

	We now consider the problem \welf, which we briefly recap: in each time step a reward $c_t \in \R$ and resource consumption vector $a_t \in [0,1]^m$ arrive, and the algorithm needs to choose $x_t \in [0,1]$ indicating the level it wants to fulfill of this request. The goal is to select the $x_t$'s in an online way to maximize the profit $$\sum_t c_t x_t - \psi\bigg(\sum_t a_t x_t \bigg).$$ In the \mixed model, in a stochastic time $t$ the information $(c_t,a_t)$ is sampled independently from an unknown distribution $\D$ (time invariant). 
	
	 To define $\OPT_{Stoch}$ in this setting, we define the optimal strategy for the stochastic part as before: let $x^*$ be the function that maps each vector $(c,v) \in \R \times [0,1]^m$ to a value $x^*(r,v) \in [0,1]$ and that maximizes $\E \sum_{t \in Stoch} c_t\, x^*(c_t,a_t) - \E \psi( \sum_{t \in Stoch} a_t\, x^*(c_t,a_t))$, and define $\OPT_{Stoch}$ as this optimum. To simplify the notation, for a stochastic time $t$, let $x^*_t := x^*(c_t,a_t)$ be the (random) optimal choice, so $\OPT_{Stoch} = \E \sum_{t \in Stoch} c_t x^*_t - \E \psi( \sum_{t \in Stoch} a_t x^*_t)$.

	We claim that in order to prove Theorem \ref{thm:welfare} it suffices to design an algorithm for the \mixed model that obtains value comparable to the $\OPT$ of only the stochastic part:
	
	\begin{thm} \label{thm:welfareStoch}
		Under the hypothesis of Theorem \ref{thm:welfare}, on the \mixed model the solution $\tilde{x}$ returned by Algorithm~\ref{alg:welfare} has expected profit at least $$\Omega(\tfrac{|Stoch|}{n})\cdot \OPT_{Stoch} - O(\psi(p \ones)).$$
	\end{thm}

	This suffices because of the following: running the adversarial $\Omega(\frac{1}{p})$-approximation of~\cite{cvxPDFOCS} over the whole instance we get at least a $\Omega(\frac{1}{p})$ fraction of the profits of the adversarial part or of the stochastic part of the instance, whichever is largest. Thus, randomly choosing between running this algorithm or Algorithm~\ref{alg:welfare} gives expected profit at least $$\Omega(\max\{\tfrac{|Stoch|}{n}, \tfrac{1}{p}\})\cdot \OPT_{Stoch} + \Omega(\tfrac{1}{p})\cdot \OPT_{Adv} - O(\psi(p\ones)),$$ as desired. \mnote{While the statement in~\cite{cvxPDFOCS} assumes that all $c_t$'s equal 1, this is without loss of generality by scaling, see the last page of \cite{cvxPDFOCS}} \mnote{Would be nice to put details in the appendix} 
	
	So we focus on proving Theorem \ref{thm:welfareStoch}. The idea of the algorithm is similar to that of the previous section: we see \welf as the minimax problem 
	\begin{align}
		\max_{(x_t) \in [0,1]^n} \bigg[\sum_t c_t x_t - \psi\bigg(\sum_t a_t x_t\bigg) \bigg] &= \max_{(x_t) \in [0,1]^n} \bigg[\sum_t c_t x_t - \sup_y \bigg(\ip{y}{\textstyle \sum_t a_t x_t} - \psi^*(y)\bigg) \bigg] \label{eq:minimax2}\\
		&= \inf_{y} \max_{(x_t) \in [0,1]^n} \bigg[\sum_t c_t x_t - \bigg(\ip{y}{\textstyle \sum_t a_t x_t} - \psi^*(y)\bigg) \bigg] \label{eq:minimax3}.
	\end{align}
	So we will replace the real cost based on $\psi$ by linearized fake costs given by $L(y,v) = \ip{y}{v} - \frac{1}{n}\, \psi^*(y)$, as before. Again we will compute approximations of the ``right slope'' $y$ online (based on the $\sup$ in \eqref{eq:minimax2}) and optimize $x_t$ based on the previous approximation $\bar{y}_{t-1}$ (based on the $\max$ in \eqref{eq:minimax3}). Let $\bar{v}_t := a_t \bar{x}_t$ be the (virtual) resource consumption incurred at time $t$. 
	
	\begin{algorithm}[H]
  \caption{Algorithm \algw for Welfare Maximization with Convex Costs}
  \begin{algorithmic}[1]
  	\For{each time $t$}
		\State (Dual) Compute $\bar{y}_t \in \R^m_+$ by feeding  $\bar{v}_1,\ldots,\bar{v}_{t-1}$ to the OCO algorithm \OCOalg
		\State 		
\parbox[t]{\dimexpr\textwidth-\leftmargin-\labelsep-\labelwidth}{(Primal) Compute the \textbf{virtual} play $\bar{x}_t$ as best response with fake cost: $$\bar{x}_t = \argmax_{x_t \in [0,1]} (c_t x_t - L(\bar{y}_t,\, a_t x_t)).$$ But \textbf{play} the scaled value $\tilde{x}_t = \frac{1}{8^2}\,\bar{x}_t$\strut}
		\EndFor
  \end{algorithmic}
  \label{alg:welfare}
	\end{algorithm}	
	
	While the details are a bit different, the high-level of its analysis is the same in the previous section. By analogy with \eqref{eq:main} we will argue that roughly
	\begin{align}
		\underbrace{\sum_t c_t \tilde{x}_t - \psi\bigg(\sum_t a_t \tilde{x}_t\bigg)}_{\textrm{\algw's profit}} \stackrel{\textrm{regret}}{\gtrsim} \underbrace{\sum_t c_t \bar{x}_t - \sum_t L(\bar{y}_t, \bar{v}_t)}_{\textrm{virtual \algw's fake profit}} \stackrel{\textrm{best resp}}{\ge} \underbrace{\sum_{t \in Stoch} c_t x^*_t - \sum_{t \in Stoch} L(\bar{y}_t, v^*_t)}_{\textrm{\OPT's fake profit}} \stackrel{\star}{\gtrsim} \OPT_{Stoch}. \label{eq:main2}
	\end{align}
		

	\subsection{Guarantee of the algorithm and Proof of Theorem \ref{thm:welfareStoch}}
	
	Recall we are assuming that $\psi = \psi_{lin} + \psi_{high}$, where $\psi_{lin}$ is a linear function and $\psi_{high}$ grows at least quadratically: $\psi_{high}(\gamma u) \ge \gamma^2\, \psi_{high}(u)$ for $\gamma \ge 1$. WLOG, assume that $\psi$ itself grows at least quadratically, otherwise run the algorithm on the instance with rewards $c - \psi_{lin}$ and cost function $\psi_{high}$.  To simplify the notation, for a stochastic time $t$ let $v^*_t := a_t x^*_t$ be $\OPT$'s resource consumption, and let $\vOPTr := \sum_{t \in Stoch} v^*_t$ be its total resource consumption. 
	
	We start proving inequality \eqref{eq:main2} from right to left. Lemma \ref{lemma:stoch} applied over the scaled input $\frac{v^*_t}{\beta}$ gives (again $\beta = \frac{n}{|Stoch|}$) 
	\begin{align*}
	\E \sum_{t \in Stoch} L\bigg(\bar{y}_t, \frac{v^*_t}{\beta}\bigg) \le \frac{1}{\beta}\, \E\,\psi\Big(\vOPT^{\,load}_{Stoch}\Big).
	\end{align*}
	Recalling that $\OPT$'s profit over the stochastic part of the instance is $\OPT_{Stoch} = \E \sum_{t \in Stoch} c_t x^*_t - \E \psi(\vOPT^{\,load}_{Stoch})$, we get
	\begin{align}
	\E\bigg[\sum_{t \in Stoch} c_t\, \frac{v^*_t}{\beta} ~- \sum_{t \in Stoch} L\bigg(\bar{y}_t, \frac{v^*_t}{\beta}\bigg)\bigg] \ge \frac{1}{\beta} \OPT_{Stoch}, \label{eq:w1}
	\end{align}
	obtaining a relationship between the fake and real profit of (scaled versions of) $\OPT$. 
			
	For the next step of inequality \eqref{eq:main2}, since $\bar{x}_t$ is a best response for which $\frac{x^*_t}{\beta}$ was a candidate, we can compare virtual \algw's and the scaled $\OPT$'s fake profits for, say, a stochastic time $t$:
	\begin{align*}
	c_t \bar{x}_t ~-  L(\bar{y}_t, \bar{v}_t) ~\ge~ c_t\, \frac{v^*_t}{\beta} - L\bigg(\bar{y}_t, \frac{v^*_t}{\beta}\bigg). 
	\end{align*}	
	We can also compare best response with the candidate $x_t = 0$ to obtain for, say, an adversarial time $t$:
	\begin{align*}
	c_t \bar{x}_t ~-  L(\bar{y}_t, \bar{v}_t) ~\ge~ - L(\bar{y}_t, 0) = \frac{1}{n}\,\psi^*(\bar{y}_t) \ge 0. 
	\end{align*}	
	Adding these bounds over the stochastic and adversarial times, respectively, we obtain a comparison between virtual \algw's fake profit and the scaled $\OPT$'s fake profits:
	\begin{align}
	\sum_t c_t \bar{x}_t ~- \sum_t L(\bar{y}_t, \bar{v}_t) ~\ge~ \sum_{t \in Stoch} c_t\, \frac{v^*_t}{\beta} ~- \sum_{t \in Stoch} L\bigg(\bar{y}_t, \frac{v^*_t}{\beta}\bigg).  \label{eq:w2}
	\end{align}	
	Finally, for the first inequality in \eqref{eq:main2}, using the definition of the scaled play $\tilde{x}_t = \frac{\bar{x}_t}{8^2}$ and the assumption that $\psi$ grows at least quadratically
	\begin{align*}
		\textrm{\algw's profit} = \sum_t c_t \tilde{x}_t - \psi\bigg(\sum_t a_t \tilde{x}_t\bigg) &= \frac{1}{8^2}\,\sum_t c_t \bar{x}_t - \psi\bigg(\frac{1}{8^2}\sum_t a_t \bar{x}_t\bigg)\\
		&\ge \frac{1}{8^2} \bigg[\sum_t c_t \bar{x}_t - \psi\bigg(\frac{1}{8} \sum_t \bar{v}_t\bigg) \bigg]
	\end{align*}
	(recall $\bar{v}_t = a_t\bar{x}_t$). The regret guarantee of algorithm \OCOalg (Item 1 of Theorem \ref{thm:OCO}) gives that the $\psi$ term on the right-hand side is at most $\sum_t L(\bar{y}_t,\bar{v}_t) + \psi(p \ones)$. Applying this to the displayed inequality gives  
	\begin{align}
		\textrm{\algw's profit} ~\ge~ \frac{1}{8^2} \bigg[\sum_t c_t \bar{x}_t - \sum_t L(\bar{y}_t, \bar{v}_t)\bigg] - \frac{1}{8^2}\, \psi(p \ones). \label{eq:w3}
	\end{align}
	
	Chaining inequalities \eqref{eq:w1}-\eqref{eq:w3} and taking expectation, we obtain 
	\begin{align*}
		\E\, \textrm{\algw's profit} \ge \frac{1}{8^2\, \beta} \,\OPT_{Stoch} - \frac{1}{8^2}\,\psi(p \ones).
	\end{align*}	
	 This concludes the analysis of the algorithm and the proof of Theorem \ref{thm:welfareStoch}.


\bibliographystyle{plain}
\bibliography{online-lp-short}


	\pagebreak
	\appendix	
	\noindent {\LARGE \bf Appendix}

	\section{$\ell_p$-norm Load Balancing} \label{app:loadBal}
	
	First, we remark that it suffices to consider $p \le \log m$, since $\|\cdot\|_{p} = \Theta(1) \cdot \|\cdot\|_{\infty}$ for all $p \ge \log m$.
	
	To obtain the guarantee from Theorem \ref{thm:loadBal}, apply the algorithm from Theorem \ref{thm:main} with the function $\psi(u) = \|u\|_p^p = \sum_i u_i^p$. Since this $\psi$ is separable and homogeneous, this yields a solution $(\bar{v}_t)_t$ that satisfies 
	\begin{align*}
		\E \bigg\|\sum_t \bar{v}_t\bigg\|_p^p ~\le~ O(1)^p\,  \bigg\|\E \sum_{t \in Stoch} v^*_t\bigg\|_p^p ~+~ O(p)^p\,\bigg\|\sum_{t \in Adv} v^*_t\bigg\|_p^p + O(1)^p\,\|p\ones\|_p^p.
	\end{align*}
Notice that the expectation in \OPT's cost in the stochastic part is inside the function $\psi = \|\cdot\|_p^p$. While this is not how the Theorem \ref{thm:main} is stated, this stronger guarantee is given in \eqref{eq:wrapHomo}.
	
	Taking $p$th root on both sides and using the subadditivity $(a + b)^{1/p} \le a^{1/p} + b^{1/p}$ that holds for all $a,b\ge 0$, we get 
	\begin{align*}
		\bigg(\E \bigg\|\sum_t \bar{v}_t\bigg\|_p^p\bigg)^{1/p} ~\le~ O(1)\,  \bigg\|\E \sum_{t \in Stoch} v^*_t\bigg\|_p ~+~ O(p)\,\bigg\|\sum_{t \in Adv} v^*_t\bigg\|_p + O(1)\,\|p\ones\|_p.
	\end{align*}
	Since $(\cdot)^{1/p}$ is concave, from Jensen's inequality we see that $\E \|\sum_t \bar{v}_t\|_p$ is at most the left-hand side of this expression. This gives the desired result.  


	\section{Proof of Lemma \ref{lemma:growth}} \label{app:fenchel}
	
	We just prove that $\ip{\nabla \psi(u)}{u} \le p \cdot \psi(u)$ for all non-negative vector $u$, and the result will follow from Lemma 4 of \cite{cvxPDFOCS}. Let $\nabla_i \psi$ denote the $i$th coordinate of $\nabla \psi$. 
		
		We first claim that $\int_0^1 \nabla_i \psi(t u)\,\d t \ge \frac{1}{p} \nabla_i \psi(u)$: For $t \in (0,1]$, using the growth assumption on $\psi$ with $\alpha = \frac{1}{t}$ we get $\nabla_i \psi(t u) \ge t^{p-1} \nabla_i \psi(u)$, and integrating over $t$ on both sides gives the claim. Thus, we have coordinate-wise $\nabla \psi(u) \le p \cdot \int_0^1 \nabla \psi(t u)\,\d t$. Using the non-negativity of $u$ and $\psi(0) = 0$, 
		\begin{align*}
			\ip{\nabla \psi(u)}{u} \le \ip{{p \cdot \textstyle \int_0^1 \nabla \psi(t u)\,\d t}\,}{u} = p \int_0^1 \ip{\nabla \psi(tu)}{u}\,\d t = p \cdot \psi(u).
		\end{align*}
		This concludes the proof. 


	\section{Proof of Theorem \ref{thm:main}: Homogeneous Case} \label{app:homo}
	
	Recall that $\psi$ is positively homogeneous of degree $p$ if for every vector $u$ and $\gamma \ge 0$ we have the equality $\psi(\gamma u) = \gamma^p \psi(u)$. We first recall a couple of facts about homogeneous convex functions that can be found, for example, in~\cite{lasserre}. 
	
	\begin{lemma}
		If $\psi : \R^m_+ \rightarrow \R_+$ is a convex function that is homogeneous of degree $p$, then its gradient $\nabla \psi$ is homogeneous of degree $p-1$.
	\end{lemma}

	As described in Section \ref{sec:homo}, consider the Lagrangian $L_{\gamma_t}(\cdot, \bar{v}_t)$ where $\gamma_t = \frac{1}{|Stoch|}$ if $t$ is a stochastic time, and $\gamma_t = 0$ otherwise. Let $(\check{y}_t)_t$ be the sequence returned by \OCOalg when run over these Lagrangian functions. Using the guarantees of Theorem \ref{thm:OCO}, we obtain a bound analogous to Corollary \ref{cor:cor} to these Lagrangians.
	
	\begin{lemma}
		We have
		\begin{align*}
			\psi\bigg(\frac{1}{8} \sum_t \bar{v}_t\bigg) ~\le~ \sum_t L_{\gamma_t}(\check{y}_t\,,\bar{v}_t) - \frac{1}{2p} \cdot \max_t \psi^*(\bar{y}_t) + \frac{3}{2} \psi(p \ones).
		\end{align*}
		Moreover, if the function $\psi$ is separable, the second term in the right-hand side can be replaced by $\frac{1}{2p} \psi^*(\bigvee_t \bar{y}_t)$.	
	\end{lemma}
	
	Moreover, using homogeneity, we claim that $\bar{v}_t$ is still a best response with respect to $L_{\gamma_t}(\check{y}_t\,,\cdot)$, so in particular:
	
	\begin{lemma}
		For all $t$, $L_{\gamma_t}(\check{y}_t\,,\bar{v}_t) \le L_{\gamma_t}(\check{y}_t\,,v_t^*)$.
	\end{lemma}
	
	\begin{proof}
		Since $\bar{y}_t$ corresponds to the Lagrangians where the $\gamma$'s (and $\bar{\gamma}$) are all equal to $\frac{1}{n}$, from equation \eqref{eq:OCOn} we see that 
		\begin{align*}
			\bar{y}_t &= \nabla \psi\bigg(\frac{4p\ones + v_{1:t-1}}{4(1 + \frac{t}{n})}\bigg)\\
			\check{y}_t &= \nabla \psi\bigg(\frac{4p\ones + v_{1:t-1}}{4(1 + \gamma_{1:t-1} + \bar{\gamma})}\bigg). 
		\end{align*}
		Since $\nabla \psi$ is homogeneous (because we assumed $\psi$ is so) we see that $\bar{y}_t$ and $\check{y}_t$ only differ by a non-negative scaling; that is, there is $\alpha_t \ge 0$ such that $\check{y}_t = \alpha_t \bar{y}_t$. Thus, since $\bar{v}_t$ is the minimizer over $V_t$ of the function $$L(\bar{y}_t, \cdot) = \ip{\bar{y}_t}{\cdot} - \frac{1}{n} \psi^*(\bar{y}_t),$$ it is clear that it also the minimizer of the function $$L_{\gamma_t}(\check{y}_t, \cdot) = \alpha_t \ip{\bar{y}_t}{\cdot} - \gamma_t \psi^*(\check{y}_t),$$ which gives the result. 
	\end{proof}
	
	In addition, the same argument as in Lemma \ref{lemma:stoch} allows us to bound the cost $L_{\gamma_t}(\check{y}_t\,,v_t^*)$, but now even more effectively because of the setting of $\gamma_t$. 
	
	\begin{lemma}
		We have
		\begin{align*}
			\E \sum_{t \in Stoch} L_{\gamma_t}(\check{y}_t, v^*_t) ~\le~ \psi\big(\E \vOPT_{Stoch}\big). 
		\end{align*}
	\end{lemma}
	
	Putting these lemmas together gives 
	\begin{align*}
		\frac{1}{8^p} \E \psi\bigg(\sum_t \bar{v}_t\bigg) \le \E \psi\bigg(\frac{1}{8} \sum_t \bar{v}_t\bigg) ~\le~ \bigg[\sum_{t \in Adv} L_{\gamma_t}(\check{y}_t\,,v^*_t)  - \frac{1}{2p} \cdot \max_t \psi^*(\bar{y}_t)\bigg] + \E \psi\big(\vOPT_{Stoch}\big) + \frac{3}{2} \psi(p \ones).
	\end{align*}
	
	The term in brackets, relative to the cost in the adversarial part of the instance, can be upper bounded in the non-separable/separable cases exactly as in Sections \ref{sec:nonsep} and \ref{sec:separable}, giving 
	\begin{align}
		\frac{1}{8^p} \E \psi\bigg(\sum_t \bar{v}_t\bigg) ~\le~ \bigg[\blue{O(p)^{2p} \textrm{ \textbf{or} } O(p)^p} \bigg] \cdot \psi \big(\vOPT_{Adv}\big)  +  \psi\big(\E \vOPT_{Stoch}\big) + \frac{3}{2} \psi(p \ones),  \label{eq:wrapHomo}
	\end{align}	
	the `\blue{\textbf{or}}' depending on the non-separable/separable case, respectively. 

	 This concludes the proof of Theorem \ref{thm:main} in the case where $\psi$ is homogeneous. 
	
\end{document}